\documentclass{article}
\usepackage{graphics}
\usepackage{amssymb}
\usepackage{latexsym}
\usepackage{rotating}
\usepackage{bibunits}
\usepackage{wrapfig}
\usepackage{psfrag}

\usepackage{a4wide}
\usepackage{amsmath}
\usepackage{url}
\usepackage{epsfig}
\usepackage{mathrsfs}
            
\newenvironment{proof}{\noindent {\bf Proof:}}{\hfill$\Box$}
\newtheorem{theorem}{Theorem}[section]
\newtheorem{lemma}[theorem]{Lemma}

\newtheorem{observation}[theorem]{Observation}

\newtheorem{definition}[theorem]{Definition}

\newcommand{\ignore}[1]{}
             
\newcommand{\hcm}[1][1]{\hspace*{#1 cm}}

\newcommand{\paren}[1]{\left( #1 \right)}

\newcommand{\ceil}[1]{\lceil #1 \rceil}

\newcommand{\Nesetril}{Ne{\u{s}}et{\u{r}}il}

\newcommand{\Ex}{\mbox{\rm Ex}}
\newcommand{\mup}{\bar{m}}
\newcommand{\mdown}{\tilde{m}}
\newcommand{\Seq}{\mathscr{S}}
\newcommand{\Rec}{\mathscr{R}}
\newcommand{\RecDeque}{\mathscr{D}}
\newcommand{\subseq}{\,\prec\,}
\newcommand{\nsubseq}{\,\nprec\,}
\newcommand{\subseqe}{\,\bar{\subseq}\,}
\newcommand{\nsubseqe}{\,\bar{\nsubseq}\,}

\title{Splay Trees, Davenport-Schinzel Sequences,\\
and the Deque Conjecture}

\author{Seth Pettie\\ The University of Michigan}

\date{}

\begin{document}
\maketitle
\begin{abstract}
We introduce a new technique to bound the asymptotic performance of splay trees.
The basic idea is to transcribe, in an indirect fashion, the rotations performed by the splay tree
as a Davenport-Schinzel sequence $\Seq$, none of whose subsequences are isomorphic to
fixed {\em forbidden subsequence}.
We direct this technique towards Tarjan's {\em deque conjecture} and prove that $n$ deque operations
require $O(n\alpha^*(n))$ time, where $\alpha^*(n)$ is the minimum number of applications of the inverse-Ackermann
function mapping $n$ to a constant.  We are optimistic that this approach could be directed towards other
open conjectures on splay trees such as the {\em traversal} and {\em split} conjectures.
\end{abstract}

\section{Introduction}

Sleator and Tarjan proposed the {\em splay tree} \cite{ST85} as a self-adjusting alternative to traditional search trees like red-black trees and AVL-trees.
Rather than enforce some type of balance invariant, the splay tree simply adjusts its structure in response to the access pattern by rotating accessed
elements towards the root in a prescribed way; see Figure~\ref{fig:zigzig-zigzag}.  
By letting the access pattern influence its own shape, the splay tree can inadvertently learn to perform
optimally on a variety of access patterns.  For example, the {\em static optimality} theorem states that splay trees are no worse than any fixed search tree.
The {\em working set}, and {\em dynamic finger} theorems show that the access time is logarithmic in the distance to the accessed 
element, where {\em distance} is either temporal (working set \cite{ST85}) or with respect to key-space (dynamic finger \cite{ColeEtal00,Cole00}).
Sleator and Tarjan went a step further and conjectured that splay trees are, to within a constant factor, just as efficient as any dynamic binary search
tree, even one that knows the whole access sequence in advance.  Despite a fair amount of attention over the years, this {\em dynamic optimality} conjecture
is still open.  In fact, there is currently no non-trivial (i.e., sub-logarithmic) bound on the competitiveness of splay trees.
The difficulty of this problem stems from the fact that splay trees were deliberately designed not to ``play along.''  They do not adhere to any notion
of {\em good structure} that we might have and, more to the point, there is no reason to believe that splay trees mimic the behavior of
the optimal dynamic search tree.

\begin{figure}[h!]
\begin{center}
\scalebox{.5}{\includegraphics{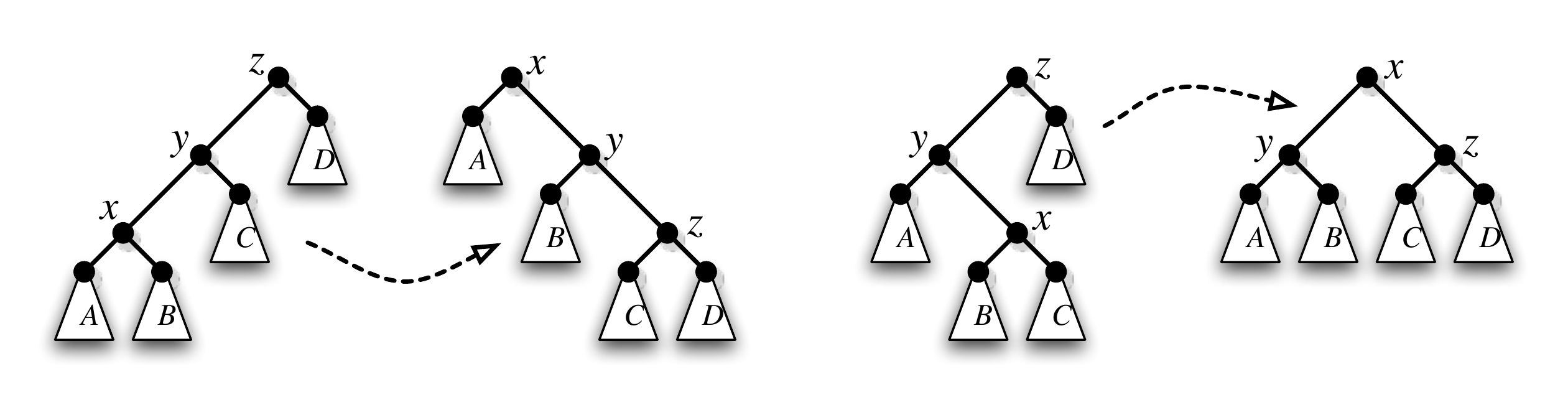}}
\end{center}
\caption{\label{fig:zigzig-zigzag}Splay trees' restructuring heuristics.  After accessing an element $x$ the tree rotates it to the root
position by repeatedly applying  a {\em zig-zig}, {\em zig-zag}, or {\em zig} as appropriate.
On the left is the zig-zig case: $x$ and its parent $y$ are both left children (or both right children); the edges
$(y,z)$ and $(x,y)$ are rotated in that order.  On the right is the zig-zag case; the edges $(x,y)$ and $(x,z)$ are
rotated in that order.  Not depicted is the zig case, when $y$ is the tree root and the edge $(x,y)$ is rotated.}
\end{figure}

The renewed interest in the dynamic optimality question is largely due to
Demaine et al.'s invention of {\em tango trees} \cite{DHIP04}.  
By appealing to the {\em interleave} lower bound of Wilbur \cite{Wilber89}
they show that tango trees are $O(\log\log n)$-competitive.  Tango trees make use of red-black trees
but it is easy to see that just about any standard binary search tree could be used as a black box in its place.
Wang et al.~\cite{WDS06} (see also \cite{Georg05}) showed that if splay trees are employed
instead of red-black trees it is possible to have $O(\log\log n)$-competitiveness and retain some properties
of splay trees, such as $O(\log n)$ amortized time per access and sequential access in linear time.  
Wang et al.~also extended their data structure to handle insertions and deletions.  

If one's immediate goal is to prove that splay trees are simply $o(\log n)$-competitive, it suffices to show that
they run in $o(n\log n)$ time on any class of access sequences for which the optimal binary search tree runs in $O(n)$ time.
There is currently no ``theory'' of access sequences whose inherent complexity is linear.  It is, therefore, not too surprising that
all the major open conjectures on splay trees (corollaries of dynamic optimality)
concern sequences whose optimal complexity is linear.  
Whether one's goal is modest or ambitious, i.e., proving sub-logarithmic competitiveness or the full dynamic optimality conjecture,
the first step must be to understand how splay trees behave on very easy access sequences.
We restate below three unresolved conjectures on the optimality of splay trees \cite{Tar85,ST85,Lucas91}.

\begin{description}
\item[Deque Conjecture] Tarjan \cite{Tar85} conjectured that all double-ended queue operations\footnote{Also called a {\em deque}, pronounced ``deck.''} 
(push, pop, and their symmetric counterparts
inject and eject) take $O(1)$ amortized time if implemented with a splay tree.  A push makes the root of the splay tree the right child of a new vertex
and a pop splays the leftmost leaf to the root position and deletes it.  Inject and eject are symmetric.

\item[Traversal Conjecture] Sleator and Tarjan \cite{ST85} conjectured that for two binary search trees $S$ and $T$ (defined on the same node set)
accessing the elements in $T$ by their preorder number in $S$ takes linear time.

\item[Split Conjecture] Lucas conjectured \cite{Lucas91} 
that any sequence of {\em splittings} in a splay tree takes linear time.  A split at $x$ consists of splaying
$x$ to the root and deleting it, leaving two splay trees, each subject to more splittings.
\end{description}

Sundar \cite{Sundar92} established a bound of $O(n\alpha(n))$ on the time required to perform $n$ deque operations, 
where $\alpha$ is the inverse-Ackermann function.  
Lucas \cite{Lucas91} showed that when the initial splay tree is a path (each node a left child), 
$n$ split operations take $O(n\alpha(n))$ time.  Notice that the split conjecture subsumes
a special case of the deque conjecture, where only pops and ejects are allowed.  We are aware of no published 
work concerning the traversal conjecture.

\paragraph{Our Contributions.}
We introduce a new technique in the analysis of splay trees that is fundamentally
different from all previous work on the subject \cite{ST85,Tar85,Sundar92,ColeEtal00,Cole00,Georg04,Elmasry04b}.
The idea is to bound the time taken to perform a sequence of accesses by 
{\em transcribing} the rotations performed by the splay tree into a Davenport-Schinzel sequence $\Seq$, i.e.,
one avoiding a specific forbidden subsequence.  We apply this technique to the deque problem
and show that $n$ deque operations take $O(n\alpha^*(n))$ time, 
where $\alpha^*(n)$ is the number of applications of the inverse-Ackermann
function mapping $n$ down to a constant.  This time bound is established by generating 
not one sequence $\Seq$ but a hierarchy of sequences, each of which avoids subsequences isomorphic
to $abababa$.  Nearly tight bounds on the length of such sequences were given by 
Agarwal, Sharir, and Shor~\cite{ASS89}.  
We believe that a generalized version of this technique should be useful
in resolving other open conjectures on splay trees.  For instance, a particularly clean
way to prove the deque, split, or traversal conjectures would be to transcribe their
rotations as a generalized Davenport-Schinzel sequence with length $O(n)$.
Klazar and Valtr~\cite{KV94} showed that a large family of forbidden subsequences have
a linear extremal function.

\paragraph{Related Work.}
Iacono \cite{Iacono05} defined a weaker notion of dynamic optimality called {\em key independent} optimality.
One assumes that keys are assigned to elements randomly.  The optimal cost of a sequence of operations
is the expected optimal cost over all key assignments.  Iacono showed that any data structure having the
working set property is also optimal in the key independent sense.  
Blum et al.~\cite{BCK03} defined another weaker notion of dynamic optimality called dynamic {\em search} optimality.
In this cost model the (online) search tree can perform any number of rotations for free after each access, i.e.,
it only pays for actually doing searches.
Georgakopoulos ~\cite{Georg04}
showed that splay trees are competitive against a large class of dynamic {\em balanced} binary search trees,
which can be ``self adjusting'' in the sense that they change shape in preparation for future searches.
Iacono defined a {\em unified property} for search structures that subsumes the working set and dynamic finger
properties.  In his data structure~\cite{Iacono01} the access time for an element is logarithmic in its distance, 
where ``distance'' is a natural combination of temporal distance plus key-space distance.  Iacono's data structure \cite{Iacono01}
is not a binary search tree and it is currently open whether {\em any} offline binary search tree has the unified property.
In other words, it is not known to be a corollary of dynamic optimality.

Just after the invention of splay trees \cite{ST85}, Fredman et al.~\cite{F+86} invented the pairing heap
as a kind of self-adjusting version of Fibonacci heaps.  There is no obvious (and still interesting) analogue of dynamic optimality
for priority queues, though Iacono~\cite{Iac00} did show that pairing heaps possess an analogue of the working set property.
See Fredman~\cite{F99} and Pettie~\cite{Pet05a} for the best lower and upper bounds on the performance of pairing heaps.

List maintenance and caching algorithms (such as move-to-front or least-recently-used \cite{ST85b}) 
are sometimes described as being self adjusting heuristics.  In these problems the (asymptotic) dynamic optimality
questions are pretty well understood \cite{ST85b}, though the leading constants have not yet been pinned down~\cite{Albers98,ASW95}.

\paragraph{Organization.} In Section~\ref{sect:equiv} we describe a known reduction \cite{Tar85,Sundar92,Lucas91} 
from the deque problem to a restrictive system of path compressions.  In Section~\ref{sect:not} we define
some notation for Davenport-Schinzel sequences, path compressions, and slowly growing functions.
Section~\ref{sect:pc} introduces a recurrence relation for a type of path compression system and shows how
it can be analyzed easily using bounds on Davenport-Schinzel sequences.  Section~\ref{sect:deque}
gives the proof that $n$ deque operations take $O(n\alpha^*(n))$ time.

\section{Deque Operations and Path Compression Schemes}\label{sect:equiv}

The relationship between deque operations on a splay tree and halving path compressions
on an arbitrary tree was noted by Lucas~\cite{Lucas90,Lucas91} and implicitly in \cite{Tar85,Sundar92}.  
Let us briefly
go through the steps of the reduction.  At the beginning of a phase we divide the $k$-node splay tree
into left and right halves of size $k/2$.  The phase ends when one half has been deleted due to pops
and ejects.   We ignore the right half for now and look at the binary tree induced by the left half; call it $L$ and its root $r$.
The root of this tree may or may not correspond to the root of the whole splay tree.  In any case, we imagine
rotating the nodes on the right spine across the root until the root has only one (left) child; call this tree $L'$.
Finally, we transform the binary tree $L'$ into a general tree $L''$ as follows.  The left
child of a vertex in $L'$ corresponds to its leftmost child in $L''$ and its right child in $L'$ corresponds
with its right sibling in $L''$.  See Figure~\ref{fig:equiv}.
\begin{figure}[h!]
\begin{center}
\scalebox{.4}{\includegraphics{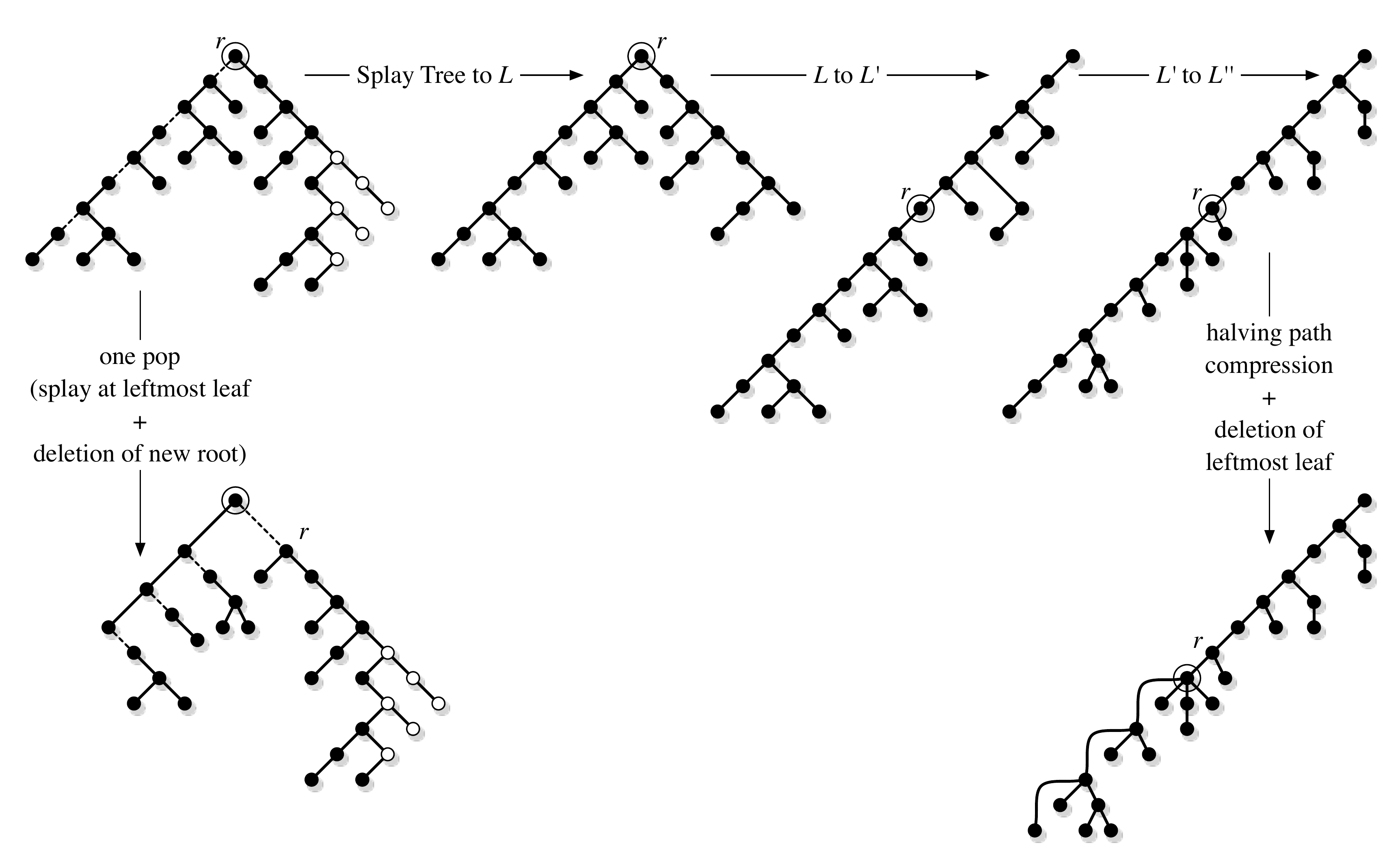}}
\end{center}
\caption{\label{fig:equiv}Top row: the transformation from the splay tree (left nodes black, right nodes right) to $L,L',$ and $L''$.
Bottom: the effect of a pop on the splay tree and $L''$.}
\end{figure}
Notice that if $v$ and its $L$-parent $p_L(v)$ lie on the left spine of $L$, and $p_L(v) \not= r$,
rotating the edge 
$(v,p_L(v))$ corresponds to making $v$ the leftmost child of its grandparent in $L''$.  If $p_L(v) = r$ rotating
$(v,p_L(v))$ makes $v$ the root of $L$ but does not change the structure of $L'$ or $L''$.  In an almost symmetric
way, if $v,p_L(v),$ and $p_L^{(2)}(v)\not= r$ lie on the right spine of $L$ then rotating $(v,p_L(v))$ in $L$
corresponds to rotating $(p_L^{(2)}(v), p_L(v))$ in $L'$ and making $p_L^{(2)}(v)$ the leftmost child of its grandparent ($v$) in $L''$.
Notice that $v = p_{L'}(p_L(v)) = p_{L'}^{(2)}(p_L^{(2)}(v))$.  Observation~\ref{obs:equiv} summarizes the relationship
between deque operations and path compressions; see also~\cite{Lucas91,Sundar92}.

\begin{observation}\label{obs:equiv}
A pop operation corresponds to a halving path compression in $L''$ that begins at the leftmost leaf and terminates at $r$,
followed by a deletion of the leftmost leaf and a possible {\em root relocation} from $r$ to its leftmost child.
A push operation causes a new leaf $v$ to be added as the leftmost child of $r$ in $L''$, followed by a root relocation from $r$ to $v$.
An eject operation corresponds to a halving path compression originating at $r$ and terminating at some ancestor (not necessarily the root of $L''$),
followed by a possible root relocation from $r$ to $p_{L''}(r)$.  An inject operation has no effect on $L''$.
\end{observation}

It is clear that if the amortized cost per deque operation in a phase is $f(k)$ then the overall cost for $m$ deque operations
on an initial $n$-node splay tree is $O((m+n)f(m+n))$; see \cite{Tar85,Sundar92,Elmasry04b}.
Using Observation~\ref{obs:equiv} we can (and do) restrict our attention to bounding the total length of a sequence of halving path
compressions up the left spine of an arbitrary rooted tree.  However, we may still lapse into deque terminology.  The terms
``pop'' and ``push'', for instance, are basically equivalent to ``halving path compression'' and ``leaf addition.''

\paragraph{Related Work.}
Restricted forms of path compressions have been studied in a number of situations.
The most well known example is Hart and Sharir's result \cite{HS86} on the equivalence between
$(ababa)$-free Davenport-Schinzel sequences and {\em generalized} postordered path compressions;
both have maximum length $\Theta(n\alpha(n))$.  Loebl and \Nesetril{} \cite{LN97} and Lucas~\cite{Lucas90} independently proved
that standard postordered path compressions with the so-called {\em rising roots} condition take linear time.
Buchsbaum, Sundar, and Tarjan~\cite{B+95} generalized this result to what they called deque-ordered path compressions,
again assuming the rising roots condition.  Hart and Sharir~\cite{HS86} have conjectured that the rising roots condition 
is not essential and that standard postordered path compressions take linear time.
The path compressions corresponding to deque operations 
are similar to the special cases studied earlier.  Some differences are that the compressions are halving (not full),
and although the compressions are spinal, due to pushes and ejects they are not performed in postorder
and do not satisfy the rising roots condition.

\section{Notation}\label{sect:not}

We say two sequences are isomorphic if they are the same up to a renaming of symbols.
The relations $\sigma \subseqe \Sigma$ and $\sigma \subseq \Sigma$ hold, respectively,
when $\sigma$ is a subsequence of $\Sigma$ and when $\sigma$ is isomorphic to a subsequence
of $\Sigma$.  A subsequence of $\Sigma$ is any sequence derived by deleting symbols from $\Sigma$.
A sequece $\Sigma$ is called $c$-regular if any two occurrences of the same symbol appear at distance at least $c$.
We denote by $|\Sigma|$ and $\|\Sigma\|$ the length and alphabet size of $\Sigma$, respectively.
Following Klazar~\cite{Klazar02} we let $\Ex(\sigma,n)$ be the maximum length of $\sigma$-free sequences:

\begin{definition}
$\Ex(\sigma,n) = \max\{|\Sigma| \;\::\:\; \sigma\nsubseq \Sigma \;\mbox{ and }\; \|\Sigma\| = n \;\mbox{ and }\; \Sigma \mbox{ is $\|\sigma\|$-regular}\}$
\end{definition}

It is known \cite{Klazar92} that $\Ex(\sigma,n) = n\cdot 2^{\alpha(n)^{O(1)}}$, where the $O(1)$ depends on the length
and alphabet size of $\sigma$.  Nearly tight bounds on $\Ex(\sigma,n)$ are known \cite{ASS89} when
$\sigma$ is of the form $ababa\cdots$.
Here $\alpha(n)$ is the inverse-Ackermann function, which can be defined as follows.
If $f(n)$ is a strictly decreasing function on the positive integers 
$f^*(n) = \min\{i \;:\; f^{(i)}(n) \le 2\}$, where $f^{(i+1)}(n) = f(f^{(i)}(n))$ and $f^{(1)}(n) = f(n)$.
Define $\alpha(m,n) = \min\{i\ge 1 \;:\; \log^{\overbrace{**\cdots*}^{i-1}}(n) \le 2+\frac{m}{n}\}$
and $\alpha(n) = \alpha(n,n)$.  This definition of $\alpha$ differs from others in the literature \cite{Tar75,HS86}
by at most a small constant.

In this paper all trees are rooted and the children of any vertex assigned some left-to-right order.
The trees we deal with are occasionally restructured with {\em path compressions}. 
Let $p(u)$ be the parent of $u$ at some specific time.  If $C = (u_1,\ldots,u_k)$ is a path,
where $u_{i+1} = p(u_i)$, performing a {\em total} compression of $C$ means to reassign
$p(u_i) = u_k$, for all $1\le i \le k-2$.  A {\em halving} compression of $C$ sets
$p(u_i) = u_{i+2}$ for all odd $i\le k-2$.  If $k$ is even a halving compression may set
$p(u_{k-1}) = p(u_k)$.  A specific case of interest is when $u_k$ is the tree root:
setting $p(u_{k-1}) = p(u_k)$ causes $u_{k-1}$ to be a tree root as well.
We say that $C$ {\em originates} at $u_1$ and {\em terminates} at $u_k$.
A total/halving path compression that does {\em not} terminate at the tree root is {\em stunted}.
The {\em length} of a total/halving path compression is the number of vertices whose
parent pointers are altered.  We never consider compressions with zero length.

A postordering of a tree rooted at $v$, having
children $v_1,\ldots,v_k$ from left to right, is the concatenation of the postorderings of
the subtrees rooted at $v_1,\ldots,v_k$ followed by $v$.  The {\em spine} of a tree is the path from
its root to its leftmost leaf.  A path compression is spinal if it affects a subpath of the spine; it need
not include the leftmost leaf nor the tree root.  It is easy to see that a total/halving spinal compression
can be postorder preserving.  For a total compression $(u_1,\ldots,u_k)$ we just prepend $u_1,\ldots,u_{k-2}$
to the preexisting left-right ordering on the children of $u_k$.  For a halving compression we make $u_{i}$ the new leftmost child
of $u_{i+1}$ for all odd $i$.

In order to use a clean inductive argument we look at a restrictive type of instance called
a {\em spinal compression system}.
The initial structure consists of a single path containing a mixture of {\em essential} nodes and
and {\em fluff} nodes.  The tree is manipulated through halving spinal path compressions, 
leaf deletions, and spontaneous compressions.   Let us
go through each of these in turn.  Whenever a fluff node becomes a leaf it is automatically deleted. 
The leftmost essential node (always a leaf) may be deleted at any time.
We are mainly interested in the total length of the ``official'' path compressions,
which are always halving and spinal.  Spontaneous compressions are any postorder preserving path compressions,
the cost of which we need not account for.
Let $\Rec(n,f,m)$ be the maximum total length of the official compressions on an instance with 
$n$ essential nodes and $f$ fluff nodes, where at most $m$ of the compressions are stunted.
In Section~\ref{sect:pc} we derive a recursive expression bounding $\Rec(n,f,m)$ and in Section~\ref{sect:deque}
we relate $\Rec$ to the time required for deque operations.

\section{Recursive Bounds on Spinal Compression Systems}\label{sect:pc}

Consider an initial instance with $n$ essential nodes and $f$ fluff nodes.  We first divide up the path
into $n/B$ blocks, each containing $B$ essential nodes, where the bottom most node in each block
is essential.  The sequence of official compressions is partitioned into $n/B$ epochs, where the $j$th epoch
begins at the first compression when the leftmost leaf belongs to the $j$th block and ends 
at the beginning of the $(j+1)$th epoch.  Let $I_j$ be the set of nodes, not including those in the $j$th block, that are touched by some compression
in the $j$th epoch.  It follows that at the commencement of the $j$th epoch $I_j$ is a single path; see Figure~\ref{fig:Ij}.
\begin{figure}[h!]
\begin{center}
\scalebox{.4}{\includegraphics{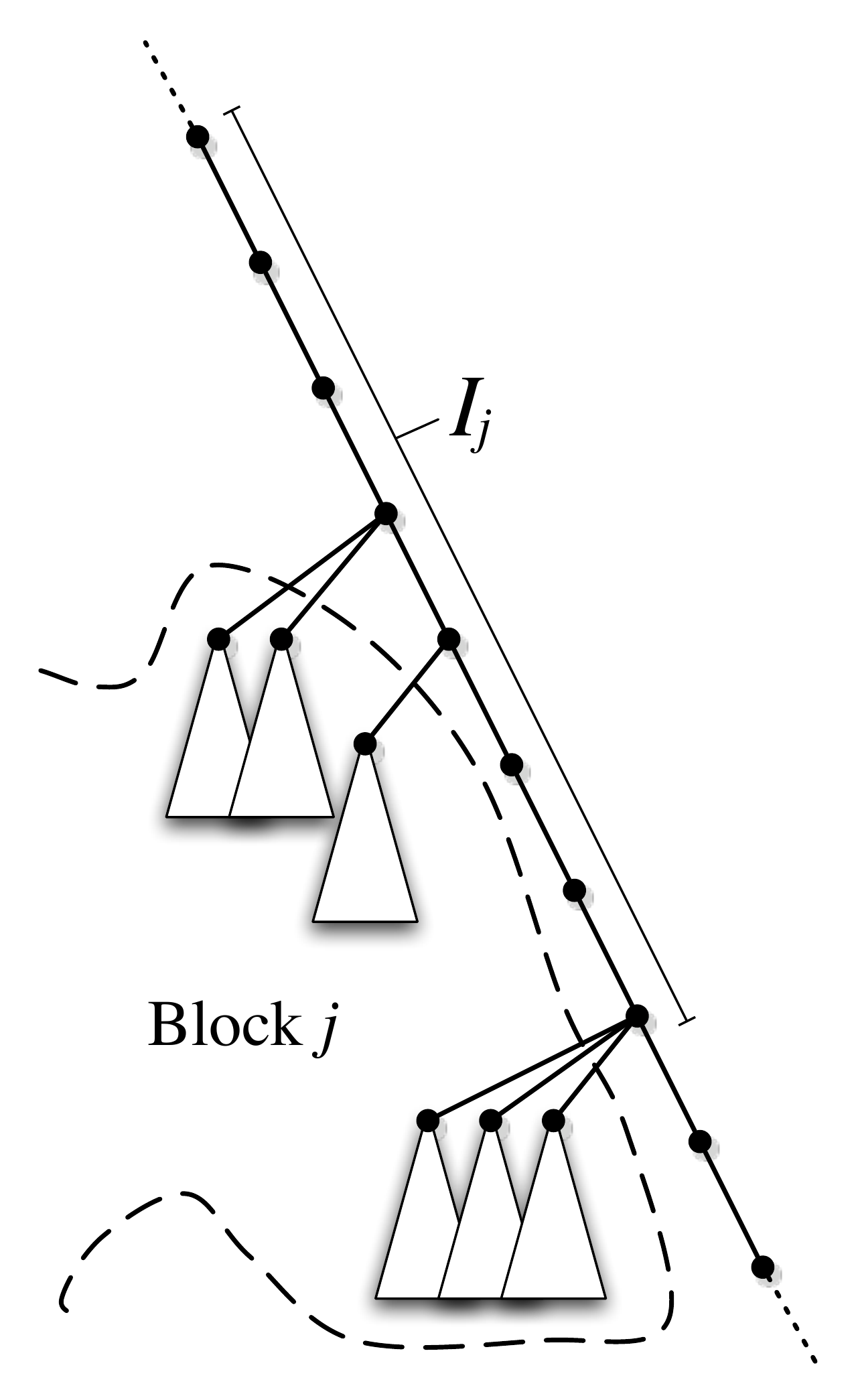}}
\end{center}
\caption{\label{fig:Ij}}
\end{figure}
At this time some members of $I_j$ may be {\em affiliated}
with previous epochs.  It is always the case that nodes affiliated with any $j' < j$ appear as a connected subpath of the spine.
We call an essential node $v$ {\em exposed} if it has no essential ancestor that either shares an affiliation with $v$ or lies
in the same block as $v$.
Let $\hat{I}_j\subseteq I_j$ be the set of essential exposed nodes.  
We call epoch $j$ sparse 
if $|\hat{I}_j|\log|\hat{I}_j| < |I_j|$ and dense otherwise.  
If epoch $j$ is dense we {\em affiliate} all nodes in $I_j$ with $j$.
We view $I_j$ as a separate spinal compression system, where $\hat{I}_j$ is the set of essential nodes and all others in $I_j\backslash \hat{I}_j$ 
are fluff.  Notice that a subpath of $I_j$ can be affiliated with $j$ and a previous epoch, say $j'$.  If a compression
influences nodes affiliated with both $j$ and $j'$ we charge the cost to the $j$th spinal compression system and call its effect
on the $j'$th compression system a spontaneous compression.
If the $j$th epoch is sparse we do not affiliate any nodes with $j$.

\begin{lemma}\label{lem:rec-a}
Let $f \ge \sum_j f_j$ and $m \ge \sum_j \mup_j + \sum_j \mdown_j$.
\[
\Rec(n,f,m) = \sum_{j=1}^{n/B} \Rec(B,f_j,\mdown_j) + \sum_{\mbox{epoch $j$ dense}} 3\Rec(|\hat{I}_j|, |\hat{I}_j|\log|\hat{I}_j|, \mup_j) + 2m + n + f
\]
\end{lemma}

\begin{proof}
Consider a compression $C$ occurring in the $j$th epoch.  
We can always break $C$ into three parts: (i) a compression $C'$ inside the $j$th block
followed by (ii) an edge $e$ leading from the $j$th block to some vertex in $I_j$, followed
by (iii) a compression $C''$ lying entirely within $I_j$.  Notice that if $C''$ is present it may be
either stunted or unstunted.  
However, if $C''$ is present then $C'$ (if present) must be an unstunted compression.
We call $C''$ {\em short} if it has unit cost, i.e., if only one node changes parent.
Whether $C'$ is stunted or not its cost is covered in the $\Rec(B,f_{j},\mdown_{j})$ term, 
where $f_{j}$ is the number of fluffy nodes and $\mdown_{j}$ the
number of compressions that are stunted with respect to block $j$.  
Similarly, if the $j$th epoch is dense the cost of $C''$, whether stunted or not, is covered in the term
$\Rec(|\hat{I}_j|, |\hat{I}_j|\log|\hat{I}_j|, \mup_j)$, where $\mup_j$ is the number of stunted compressions that terminate in 
the spinal compression system on $I_j$.  We account for the cost of $e$ in one of two ways.  If $C''$ (and therefore $C$)
is stunted we charge it to the compression; hence the $m$ term.  If not, notice that after every unstunted compression $C''$
the number of nodes whose grandparents are roots or nonexistent increases by at least one.  This can obviously happen
at most $n+f$ times.
Consider the compression $C''$ when the $j$th epoch is sparse.  If $C''$ is short we charge its cost to the compression; hence the other $m$ term.
In general, $C''$ will intersect one or more previously spawned spinal compression systems.  That is, it can be divided up
into pieces $C''_1, C''_2, C''_3,\ldots$, where for odd $i$, $C''_i$ intersects an existing compression system.
The costs of $C''_1, C''_3,\ldots$ are covered in the $\sum_{j'} \Rec(|\hat{I}_{j'}|, |\hat{I}_j|\log|\hat{I}_j|, \mup_{j'})$ terms.
By doubling this sum we can account for some of $C''_2,C''_4,\ldots$; at the very least this includes those with unit cost.
Over all of epoch $j$, the total length of the other mini-compressions of the form $C''_i$ ($i$ even)
is $|\hat{I}_{j}|\log|\hat{I}_j|$;  see Sundar~\cite{Sundar92}.\footnote{In Sundar's terminology each of these mini-compressions is a $k$-cascade for some $k\ge 2$.}
By tripling the sum $\sum_{j'} \Rec(|\hat{I}_{j'}|, |\hat{I}_{j'}|\log|\hat{I}_{j'}|, \mup_{j'})$ we account for the cost of these mini-compressions as well.
\end{proof}

Lemma~\ref{lem:rec-a} looks as though it may be effectively vacuous.
We are trying to bound the growth of $\Rec(n,f,m)$ and all Lemma~\ref{lem:rec-a} can say 
is that it depends on the magnitude of the $\{\hat{I}_j\}_j$ sets.  It does not even suggest a trivial
bound on their size.   Our strategy is to transcribe the $\{\hat{I}_j\}_j$ sets as a repetition-free
sequence $\Seq$ that avoids a specific forbidden subsequence.  It follows from the results of Agarwal et al.~\cite{ASS89}
that $|\Seq|$ is nearly linear in the size of its alphabet.
By choosing an appropriate block size $B$ we can apply Lemma~\ref{lem:rec-a} to obtain useful bounds on $\Rec$.

Our transcription method for $\Seq$
is very similar to the one used by Hart and Sharir \cite{HS86}.  
In Lemma~\ref{lem:abaabba} we show that $abaabba\nsubseq \Seq$ and $abababa\nsubseq \Seq$.
Using bounds on $\Ex(abababa,\cdot)$ from Agarwal, Sharir, and Shor~\cite{ASS89} we 
are able to show, in Lemmas~\ref{lem:rec-b} and \ref{lem:alphastar} that $\Rec(n,f,m)$ is $O((n+f+m)\alpha^*(n))$.

Before giving the actual transcription method we give a simpler one, point out why it 
isn't quite what we need, then adjust it.  The transcription is based on an evolving labeling
of the nodes.  A label is simply a sequence of block/epoch indices, listed in {\em descending} order.
Since nodes can belong to several spinal compression systems a node may keep several labels,
one for each system.
At the commencement of the $j$th epoch we
select out of $I_j$ a subset $\hat{I}_j$ with several properties, one of which is that
$\hat{I}_j$ has at most one node from any block.  If the $j$th epoch is dense
we prepend $j$ to the label of each node in $\hat{I}_j$.  The sequence $\Seq'$ consists
of the concatenation of all node labels, where the nodes are given in postorder.  An equivalent definition
is that $\Seq'$ begins empty; whenever the leftmost leaf is deleted we append its label to $\Seq'$ and continue.
Besides having the property that $abaabba, abababa\nsubseq \Seq'$, we need
$\Seq'$ to be repetition-free and contain not too many occurrences of any one symbol,
say, at most $t$ occurrences.  To enforce this, if $|\hat{I}_j| > t$ we simply split it up into $\ceil{|\hat{I}_j|/t}$ pieces
and treat each piece as a distinct spinal compression system.  Thus, the number of systems could 
exceed the number of blocks/epochs.  There may be repetitions in $\Seq'$ but not too many.  Note that the labels
of nodes appearing in any block share no symbols in common.  Therefore, $\Seq$ can only contain repetitions at block boundaries.
By removing at most $n/B - 1$ symbols from $\Seq'$ we can eliminate all repetitions.  
Using the above two observations, let $\Seq$ be a repetition-free sequence derived from $\Seq'$ in 
which no symbol occurs more than $t$ times.

\begin{lemma}\label{lem:abaabba}
For $\sigma\in\{abaabba,abababa\}$, $\sigma \nsubseq \Seq'$.
\end{lemma}

\begin{proof}
First note that if the lemma holds for $\Seq'$ it holds for $\Seq$ as well.
We show that for any $b>a$, $babba\nsubseqe \Seq'$. 
This would prove the lemma since, for any $\sigma' \in \{babba,abaab\}$ and $\sigma'' \in \{abaabba, abababa\}$, 
$\sigma'\subseqe \sigma''$.
Consider the commencement of the $b$th epoch.  If no nodes affiliated with $a$ appear on the path $I_b$
then all nodes labeled with $b$ occur, in postorder, strictly before or strictly after all nodes labeled with $a$; see the left part of 
Figure~\ref{fig:abaabba}.  
\begin{figure}[h!]
\begin{center}
\scalebox{.4}{\includegraphics{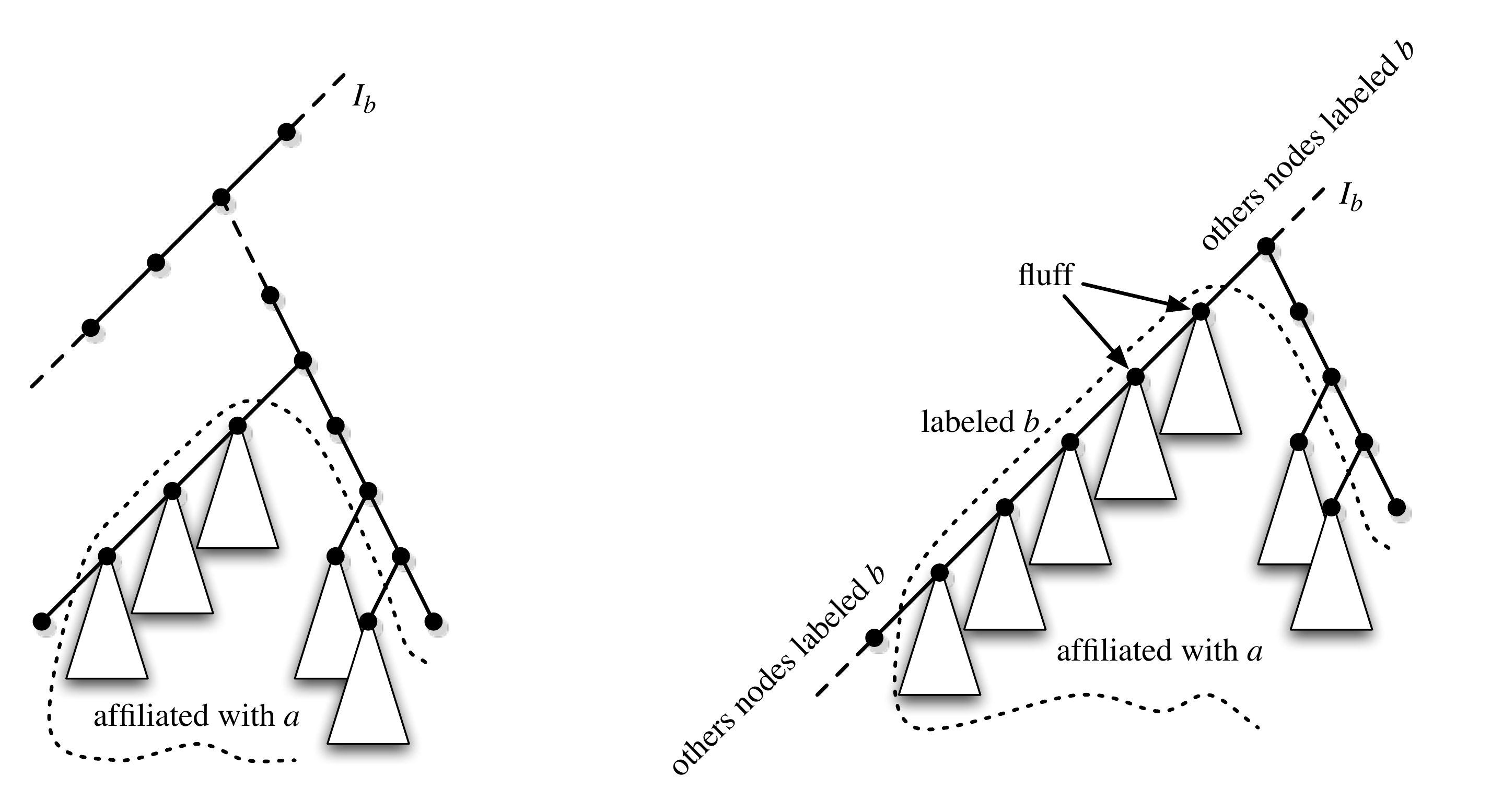}}
\end{center}
\caption{\label{fig:abaabba}}
\end{figure}
In this case
$baba \nsubseqe \Seq'$.  The more interesting case is when some interval of the nodes in $I_b$ are affiliated with $a$.
Recall that in our transcription method only exposed nodes were labeled and there could be at most one exposed node
in each set of nodes with a common affiliation.
Thus, only one node affiliated with $a$ can be labeled $b$.   All other nodes labeled
$b$ occur strictly before or strictly after all nodes labeled $a$.  Thus, after the appearance of $babb$ in $\Seq'$ all nodes
labeled $a$ would have been deleted.  We conclude that $babba \nsubseqe \Seq'$.
\end{proof}

Lemma~\ref{lem:rec-b} incorporates the recursive characterization of $\Rec$ from Lemma~\ref{lem:rec-a}
and the $(abababa)$-freeness of $\Seq$ established in Lemma~\ref{lem:abaabba}.

\begin{lemma}\label{lem:rec-b}
\[
\Rec(n,f,m) = \sum_{j=1}^{n/B} \Rec(B,f_j,\mdown_j) + 3\sum_{i=1}^{q} \Rec(l_i, l_i\log l_i, 2l_i\log^2 l_i) + 2m + n + f
\]
where 
\[
\begin{array}{l@{\hcm[.5]}l@{\hcm[.5]}l@{\hcm[.5]}l@{\hcm[.5]}l}
f \ge \sum_j f_j, &
m \ge \sum_j \mdown_j, &
l_i \le t, &
q \le \frac{\sum_i l_j}{t} + \frac{n}{B}, &
\sum_i l_i \le \Ex(abababa, q)
\end{array}
\]
\end{lemma}

\begin{proof}
Here $\{l_i\}_i$ represent the sizes of the spawned spinal compression systems.  They would 
correspond with the epochs/blocks if we did not artificially break them up to guarantee that each $l_i \le t$.
It follows that the number of spawned systems is at most $q = (\sum_i l_i)/t + n/B$.  By Lemma~\ref{lem:abaabba}
we have $\sum_i l_i \le \Ex(abababa,q)$.  The term $\Rec(l_i, l_i\log l_i, 2l_i\log^2 l_i)$ reflects the cost of the $i$th spawned compression system,
with $l_i$ essential nodes and at most $l_i\log l_i$ fluff nodes.  By \cite{Sundar92} the number of stunted compressions (with greater than unit cost)
terminating in this system is at most $(l_i + l_i\log l_i)\log(l_i + l_i\log l_i) \le 2l_i\log^2 l_i$.
\end{proof}

\begin{lemma}\label{lem:alphastar}
$\Rec(n,f,m) = O((n+f+m)\alpha^*(n))$
\end{lemma}

\begin{proof}
Let $\min\{\Ex(abaabba,z), \Ex(abababa, z)\} = z\cdot \beta(z)$.  
From the bounds established by Hart and Sharir \cite{HS86}
and Agarwal et al.~\cite{ASS89} we only know that $\beta(z) = \Omega(\alpha(z))$ and $O(\alpha(z))^{\alpha(z)}$.\footnote{Klazar's 
and Valtr \cite{KV94} claimed that $\Ex(abaabba,n) = \Theta(n\alpha(n))$ could be had by fiddling with Hart and Sharir's proof \cite{HS86}.
This claim was later retracted \cite{Klazar02}.}
We apply Lemma~\ref{lem:rec-b} with
$B = t = \beta^2(n)$:
\begin{eqnarray*}
\sum_i l_i &\le& \min\{\Ex(abababa, q), \Ex(abaabba,q)\}\\
               &\le&  \min\{\Ex(abababa, \frac{\sum_j l_j}{t} + \frac{n}{B}), \Ex(abaabba,\frac{\sum_j l_j}{t} + \frac{n}{B})\}\\
               &\le& (\frac{\sum_i l_i}{t} + \frac{n}{B})\beta(n)\\
               &\le& \frac{\sum_i l_i}{\beta(n)} + \frac{n}{\beta(n)}
\end{eqnarray*}

Thus $\sum_j l_j \le n/(\beta(n)-1)$.  Assume inductively that $\Rec(n,f,m) \le c(n+f+m)\gamma(n)$, where 
$\gamma(n)$ is the minimum $i$ such that $(\beta^2)^{(i)}(n)$ is less than some large enough constant.  (We need 
this constant threshold since $\beta^2(n)$ is not necessarily decreasing for very small $n$.)

\begin{eqnarray*}
\Rec(n,f,m) &\le& \sum_{j=1}^{n/B} \Rec(B,f_j,\mdown_j) + \sum_{i=1}^{q} 3\Rec(l_i, l_i\log l_i, 2l_i\log^2 l_i) + 2m + n + f\\
                 &\le& c(n+f+\mbox{$\sum_j \mdown_j$})\gamma(\beta(n)^2) + 3c\paren{\frac{n}{\beta(n)-1}(1+\log(\beta^2(n)) + 2\log^2(\beta^2(n)))}\gamma(\beta^2(n)) + 2m + n+f\\
                 &\le& c(n+f+\mbox{$\sum_j \mdown_j$})(\gamma(n)-1) + o(n) + 2m + n+f\\
                 &\le& c(n+f+m)\gamma(n) \hcm[1.5]\{\mbox{for $c$ sufficiently large}\}\\
\end{eqnarray*}
It is easy to see that $\beta^2(\beta^2(n)) < \alpha(n)$ for $n$ sufficiently large, from which it follows that
$\gamma(n) = O(\alpha^*(n))$.
\end{proof}

\section{Upper Bounds on Deque Operations}\label{sect:deque}

\begin{theorem}
Starting with an $n$-node splay tree, the time required for $m$ deque operations
is $O((m+n)\alpha^*(m+n))$.
\end{theorem}

\begin{proof}
We reduce the deque problem to a set of spinal compression systems.  The main difference
between these systems and the deque problem as a whole is that spinal compression systems
do not allow general insertion/deletion of leaves and the initial tree is always a single path.

We impose a linear order on all nodes that
ever appear in the splay tree.  A pushed node (injected node) precedes (proceeds)
all nodes that are currently in or were in the splay tree in the past.  
As in the general reduction from Section~\ref{sect:equiv} we divide the access sequence into phases.
At the beginning of each phase the nodes are separated into equally sized left and right sets, say $m'$ in each set.
We partition the left set into blocks of size $B= \beta(m')^2$.
Let the $j$th period begin at the first pop of an element in the $j$th block and end
when all of the $j$th block has been popped.  Note that the pop that begins the $j$th period
could delete {\em any} element from the $j$th block, not necessarily the first.
Also note that periods are not necessarily disjoint.  For example, just after the $j$th period begins
we could push $B$ elements and then perform a pop.  This would have the effect of starting the
$j'$th period without ending the $j$th period, where $j'>j$.  Only one period is {\em active}, namely, the 
one whose block contains the leftmost leaf in the tree.  All other periods are {\em on hold}.

Let $J_j$ be the {\em set} of nodes that participate in a compression in the active part of the $j$th period,
excluding those that lie in block $j$.  Clearly the nodes in $J_j$ lie on one path when the $j$th period begins.
However, they are not necessarily contiguous.  For example, if the $j$th period is put on hold, compressions
from later periods could evict non-$J_j$ nodes from the path between two $J_j$ nodes.
If we put ourselves in one period's point of view we can assume without loss of generality that $J_j$ {\em does}
form a contiguous segment since, by definition, this period never ``sees'' any non-$J_j$ nodes that would expose 
its misbelief.

Let $\RecDeque(m')$ be the maximum cost of deque operations (i.e., the cost of push, pop, and eject operations that affect the left set)
when the total number of nodes is $m'$.
After the $j$th period begins the $j$th block functions as a mini-deque structure and its total cost at most $\RecDeque(B)$.
Until the $j$th period begins the block-$j$ nodes may participate in compressions, the total cost of which is $B$.
(Whenever a node is evicted from the left spine of the tree induced by block-$j$ nodes it can only be seen again
after the $j$th period begins.)
The cost of compressions inside the $\{J_j\}_j$ sets can be expressed in terms of the $\Rec$ function.
Again, at the commencement of the $j$th period we identify the exposed nodes $\hat{J}_j\subseteq J_j$.
We do nothing if it is sparse ($|\hat{J}_j|\log|\hat{J}_j| < |J_j|$) and if it is dense, we affiliate all nodes in $J_j$ with $j$ and spawn a separate
spinal compression system on $J_j$.
The proofs of Lemmas~\ref{lem:rec-a} and \ref{lem:rec-b} can easily be adapted to show that:
\[
\RecDeque(m') \le \frac{m'}{B}\RecDeque(B) + O(m' + \sum_{j=1}^{m'/B} \Rec(|\hat{J}_j|, |\hat{J}_j|\log|\hat{J}_j|, 2|\hat{J}_j|\log^2 |\hat{J}_j|))
\]
\noindent where $\sum_{j} |\hat{J}_j| \le \Ex(abababa, m'/B) = (m'/B)\beta(m'/B) \le m'/\beta(m')$.
A simple proof by induction (along the lines of Lemma~\ref{lem:alphastar}) shows that $\RecDeque(m') = O(m'\alpha^*(m'))$.
\end{proof}

\section{Conclusion}\label{sect:conclusion}

It's fair to say that previous analyses of splay trees could be characterized
as using potential functions~\cite{ST85,Georg04}, counting 
arguments~\cite{Tar85,Lucas91,Sundar92,Elmasry04b} or a complex synthesis 
of the two~\cite{ColeEtal00,Cole00}.  Although these techniques have been wildly successful
in proving (or nearly proving) the {\em corollaries} of dynamic optimality,
they have had no impact on the dynamic optimality question itself.  
The reason for this is probably the fact that so little is known about the behavior
of the optimal (offline) dynamic binary search tree; see \cite{Lucas88,BCK03}.
It is worth noting that the $O(\log\log n)$-competitiveness of tango trees and their
variants \cite{DHIP04,WDS06,Georg05} was established not by comparing them against
an optimal tree but one of Wilber's lower bounds on any binary search tree~\cite{Wilber89}.
(Contrast this technique with the more direct approach used by Sleator and Tarjan~\cite{ST85b}
in their proof of the dynamic optimality of Move-To-Front.)

The strategy taken in this paper is quite different from previous work and is clearly general
enough to be applied to other open problems concerning splay trees.
By transcribing the rotations performed by the splay tree into a Davenport-Schinzel sequence
avoiding $abaabba$ and $abababa$ we were able to use an off-the-shelf result of Agarwal et al.~\cite{ASS89}
and avoid an unmanageable bookkeeping problem.
One direction for future research is to study the relationship between splay operations,
various transcription methods, and various forbidden subsequences.  
Of particular interest are those forbidden subsequences whose extremal function is linear
since it is these that would be fit for finally settling the deque, split, and traversal conjectures.
See Adamec et al.~\cite{AKV92} and Klazar and Valtr~\cite{KV94}  for a large family of 
linear forbidden subsequences.

{\small
\bibliographystyle{plain}
\bibliography{../references}
}

\end{document}